\documentclass[conference]{IEEEtran}
\IEEEoverridecommandlockouts
\usepackage{cite}
\usepackage{amsmath,amssymb,amsfonts}
\usepackage{amsthm}
\usepackage{algorithmic}
\usepackage{graphicx}
\usepackage{textcomp}
\usepackage{xcolor}
\def\BibTeX{{\rm B\kern-.05em{\sc i\kern-.025em b}\kern-.08em
    T\kern-.1667em\lower.7ex\hbox{E}\kern-.125emX}}
\usepackage{graphicx}
\usepackage{bm}

\theoremstyle{plain}

\newtheorem{theorem}{Theorem}
\newtheorem{proposition}{Proposition}
\newtheorem{corollary}{Corollary}
\newtheorem{definition}{Definition}

\newtheorem{claim}{Claim}

\newcommand{\mO}{{O}}
\newcommand{\bp}{\bm{p}}

\newcommand{\cw}{{\mathtt{cw}}}

\newcommand{\nd}{{\mathtt{nd}}}

\newcommand{\mw}{{\mathtt{mw}}}

\newcommand{\dist}{\mathrm{dist}}

\begin{document}

\title{Solving Distance-constrained Labeling Problems \\for Small Diameter Graphs via TSP\textsuperscript{*}
\thanks{\textsuperscript{*}This work is partially supported by JSPS KAKENHI Grant Number JP20H05967, JP21H05852, JP21K19765, JP21K17707, JP22H00513.}
}

\author{\IEEEauthorblockN{Tesshu Hanaka}
\IEEEauthorblockA{\textit{Department of Informatics} \\
\textit{Kyushu University}\\
Fukuoka, Japan \\
{hanaka@inf.kyushu-u.ac.jp}}
\and
\IEEEauthorblockN{Hirotaka Ono}
\IEEEauthorblockA{\textit{Department of Mathematical Informatics} \\
\textit{Nagoya University}\\
Nagoya, Japan \\
ono@nagoya-u.jp}
\and
\IEEEauthorblockN{Kosuke Sugiyama}
\IEEEauthorblockA{\textit{Department of Mathematical Informatics} \\
\textit{Nagoya University}\\
Nagoya, Japan \\
sugiyama.kousuke.k3@s.mail.nagoya-u.ac.jp}}

\maketitle

\begin{abstract}
For an undirected graph $G=(V,E)$ and a $k$-non-negative integer vector $\bp=(p_1,\ldots,p_k)$, a mapping $l\colon V\to \mathbb{N}\cup \{0\}$ is called an $L(\bp)$-labeling of $G$ if $\left| l(u)-l(v) \right|\geq p_d$ for any two distinct vertices $u,v\in V$ with distance $d$, and the maximum value of $\{l(v)\mid v\in V\}$ is called the span of $l$. 
Originally, $L(\bp)$-labeling of $G$ for $\bp=(2,1)$ is introduced in the context of frequency assignment in radio networks, where ‘close’ transmitters must receive different frequencies and ‘very close’ transmitters must receive frequencies that are at least two frequencies apart so that
they can avoid interference.
\textsc{$L(\bp)$-Labeling} is the problem of finding the minimum span $\lambda_{\bp}$ among $L(\bp)$-labelings of $G$, which is NP-hard for every non-zero $\bp$. 
\textsc{$L(\bp)$-Labeling} is well studied for specific $\bp$'s; in particular, many (exact or approximation) algorithms for general graphs or restricted classes of graphs are proposed for $\bp=(2,1)$ or more generally $\bp=(p,q)$. Unfortunately, most algorithms strongly depend on the values of $\bp$, and it is not apparent to extend algorithms for $\bp$ to ones for another $\bp'$ in general. 
In this paper, we give a simple polynomial-time reduction of \textsc{$L(\bp)$-Labeling} on graphs with a small diameter to \textsc{Metric (Path) TSP}, which enables us to use numerous results on \textsc{(Metric) TSP}. 
On the practical side, we can utilize various high-performance heuristics for TSP, such as Concordo and LKH, to solve our problem. On the theoretical side, we can see that the problem for any $\bp$ under this framework is 1.5-approximable, 
and it can be solved by the Held-Karp algorithm in $O(2^n n^2)$ time, where $n$ is the number of vertices, and so on. 
\end{abstract}


\section{Introduction}
For an undirected graph $G$ with $n$ vertices and $m$ edges, and a $k$-nonnegative integer vector $\bp=(p_1,\ldots,p_k)$, a mapping $l\colon V\to \mathbb{N}\cup \{0\}$ is an $L(\bp)$-labeling of $G$ if $\left| l(u)-l(v) \right|\geq p_d$ for any two distinct vertices $u,v\in V$ with distance $d$, 
the maximum value of $\{l(v)\mid v\in V\}$ is called the span of $l$. 
The notion of $L(\bp)$-labeling for $\bp=(2,1)$ can be seen in Hale~\cite{H80} and Roberts~\cite{R91}
in the context of frequency assignment in radio networks, where ‘close’ transmitters must receive different frequencies
and ‘very close’ transmitters must receive frequencies that are at least two frequencies apart so that
they can avoid interference.
\textsc{$L(\bp)$-Labeling} is the problem of finding the minimum span $\lambda_{\bp}$ among $L(\bp)$-labelings of $G$, which is NP-hard for every non-zero $\bp$. 
Since \textsc{$L(\bp)$-Labeling} for $k=1$ is the ordinary coloring problem, the cases of $k\ge 2$ are essential to study \textsc{$L(\bp)$-Labeling} under its name. In particular, the problem for $\bp=(p_1,p_2)=(p,q)$ is called the \textsc{$L(p,q)$-Labeling} problem and intensively and extensively studied. 

Among infinite settings of $(p,q)$, probably \textsc{$L(2,1)$-Labeling} is most studied. It is shown that \textsc{$L(2,1)$-Labeling} is NP-hard even for restricted 
classes of graphs, such as planar graphs, bipartite graphs, chordal graphs~\cite{BKTL04}, graphs with diameter $2$~\cite{griggs1992labelling}, and graphs of tree-width 2~\cite{fiala2005distance}. In contrast, only a few graph classes are known to be solvable in polynomial time. For example, \textsc{$L(2,1)$-Labeling} can be solved in polynomial time for paths, cycles, wheels~\cite{griggs1992labelling}, co-graphs, and trees~\cite{Chang1996,Hasunuma2013:tree}. These algorithms are straightforward (paths, cycles, wheels) or strongly depend on the properties of graphs (co-graphs and trees). In fact, the NP-hardness for graphs of tree-width 2 implies that the polynomial-time solvability for trees (graphs of tree-width 1) depends on not a tree-like structure but the tree structure itself; it might be difficult to extend or generalize algorithms for trees to superclasses of trees. Note that the algorithm of \cite{Hasunuma2013:tree} for trees is quite involved though its running time is linear. Furthermore, \textsc{$L(p,q)$-Labeling} is NP-hard even for trees, if $p$ and $q$ do not have a common divisor.  

Another direction of research for intractable problems is to design exact exponential-time algorithms whose bases or exponents are small. For example, Junosza-Szaniawski et al.~\cite{junosza2013fast} present an algorithm for \textsc{$L(2,1)$-Labeling} whose running time is $\mO(2.6488^n)$, which is currently the fastest. This algorithm uses the exponential size of memories. 
The current fastest exact algorithm with polynomial space for
\textsc{$L(2,1)$-Labeling} is proposed by Junosza-Szaniawski et al.~\cite{junosza2013determining}, and it runs in $\mO(7.4922^n)$ time. These algorithms are specialized in \textsc{$L(2,1)$-Labeling}. 
As more generalized algorithms, Cygan and Kowalik presented an exact algorithm for a more general labeling problem, called \emph{channel assignment problem}. It is based on the fast zeta transform in combination with the inclusion-exclusion principle~\cite{cygan2011channel}. The algorithm solves \textsc{$L(p,q)$-Labeling} in $\mO^*((\max\{p,q\}+1)^n)$ time and \textsc{$L(2,1)$-Labeling} in $\mO^*(3^n)$ time, where polynomial factors are omitted in $\mO^*$ notation.  

In summary, \textsc{$L(\bp)$-Labeling} is well studied in the fields of algorithm design, but most of the developed algorithms are tailored to $\bp$ and graph classes, and it is hard to generalize them.   

\subsection{Our contribution}
In this paper, we address the \textsc{$L(\bp)$-Labeling} problem on graphs with a small diameter, which is known to be NP-hard. Our approach is simple; we just solve the problem via TSP. Namely, our main contribution is an $O(nm)$-time reduction from \textsc{$L(\bp)$-Labeling} for graph $G$ with diameter at most the dimension of $\bp$, say $k$,  
to \textsc{Metric Path Traveling Salesman Problem (TSP)} under the assumption that $p_{\max}\le 2p_{\min}$, where $p_{\min}=\min \{p_1,\ldots, p_k\}$ and $p_{\max} = \max\{p_1,\ldots,p_k\}$. 
%
%
Note that the most well-studied setting $\bp=(2,1)$ satisfies this condition. Although this reduction is available only for graphs with a small diameter and $\bp$ satisfying the above condition, it enables us to use numerous results of \textsc{(Metric) TSP}. 

On the practical side, since many practical algorithms for (Metric) TSP have been developed, they can be applied to solve \textsc{$L(\bp)$-Labeling} for graphs with a small diameter with a minor modification. For example, the Lin-Kernighan heuristic for symmetric TSP~\cite{lin1973effective} and its variants are known to have outstanding performance, and there are several excellent implementations~\cite{Concorde,helsgaun2000effective}. Such implementations can be used to solve our problems as engines practically.  

On the theoretical side, the reduction leads to several algorithms with performance guarantees, such as an $O(2^n n^2)$-time algorithm and a 1.5-approximation algorithm for \textsc{$L(\bp)$-Labeling} if the diameter of an input graph is at most $k$ and if $p_{\max}\leq 2p_{\min}$. Both of the results imply that a small diameter and the setting $\bp$ may make the problem easier;  it is only known that \textsc{$L(p,q)$-Labeling} for general graphs can be solved in $O^*((\max\{p,q\}+1)^n)$ time and be $O(\min\{\Delta,\sqrt{n} + p/q\})$-approximable in polynomial time, where $\Delta$ is maximum degree. Particularly, in case of $k=2$, our reduction reduces the problem (i.e, \textsc{$L(p,q)$-Labeling}) to Path TSP with 2-valued edge weights, which can be solved via \textsc{Partition into Paths}. Since \textsc{Partition into Path} is known to be fixed-parameter tractable for modular-width~\cite{gajarsky2013parameterized}, so is our problem. 
On the other hand, we point out that \textsc{$L(p,q)$-Labeling} for graphs with diameter 2 is W[1]-hard for clique-width, which could show a frontier between fixed parameter (in)tractability. 

In passing, we can show that \textsc{$L(1,\ldots,1)$-Labeling} on general graphs is fixed-parameter tractable for modular-width. Although the parameterized complexity of \textsc{$L(p,q)$-Labeling} for modular-width remains open in general, \textsc{$L(\bp)$-Labeling} becomes $p_{\max}$-approximable in FPT time for modular-width by the FPT result for \textsc{$L(1,\ldots,1)$-Labeling}.


\subsection{Related work}
\subsubsection{Distance-Constrained labeling}
The original notion of distance-constrained labeling can be seen in Hale \cite{H80} and Roberts \cite{R91}  in the context of frequency assignment. In frequency assignment, `close' transmitters must receive different frequencies, and `very close' transmitters must receive frequencies that are at least two frequencies apart to avoid interference.  
Then, Griggs and Yeh formally introduced the notion of $L(p,q)$-labeling in~\cite{griggs1992labelling}.  
Since $p$ and $q$ could be any natural numbers, there are infinite settings of $L(p,q)$-labeling, but $L(2,1)$-labeling is most studied. One of the reasons is the context of more general frequency assignment because the setting explained above is interpreted as $L(2,1)$-labeling. In the context of frequency assignment, it is natural to consider the setting of $p\ge q$. Also, $q=1$ might be natural because it decides the unit. Another reason why $L(2,1)$ is most popular is that the setting of $p=2$ and $q=1$ seems the most natural and fundamental among the settings represented by $L(p,q)$-labeling. Indeed, $L(1,1)$-labeling of $G$ is equivalent to the ordinary coloring on the square of $G$; we do not need to study $L(1,1)$-labeling itself in this name.  

As introduced in the previous sections, the \textsc{$L(p,q)$-labeling} problem or specifically the \textsc{$L(2,1)$-labeling} problem is NP-hard even for restricted classes of graphs. Thus polynomial-time algorithms for particular classes of graphs and exact exponential-time algorithms are developed. We list here other results than those mentioned before. As for approximation, \textsc{$L(p,q)$-Labeling} is NP-hard to approximate within factor better than $n^{\frac{1}{2}-\varepsilon}$.  On the other hand, there is an asymptotically tight $O(\min\{\Delta,\sqrt{n} + p/q\})$-approximation algorithm where $\Delta$ is the maximum degree of $G$~\cite{halldorsson2006approximating}.

For the parameterized complexity, the \textsc{$L(2,1)$-Labeling} problem is fixed-parameter tractable for vertex cover number~\cite{Fiala2011}, clique-width plus maximum degree, or twin cover number plus maximum clique size~\cite{Hanaka2022}.
Although it is less critical to study \textsc{$L(1,\dots,1)$-Labeling} (we write \textsc{$L(\mathbf{1})$-Labeling} hereafter) in this name, $L(\mathbf{1})$-Labeling can be used for approximating \textsc{$L(\bp)$-Labeling};
$L(\mathbf{1})$-Labeling yields $p_{\max}$-approximation of $L(\bp)$-Labeling, $p_{\max} = \max_{d\in [k]} p_d$. 
For this reason, we are interested in the complexity of 
\textsc{$L(\mathbf{1})$-Labeling} or Coloring of powers of graphs. It is known that \textsc{$L(1,1)$-Labeling} is W[1]-hard for the tree-width~\cite{Fiala2011}, even though the ordinary Coloring is FPT, but \textsc{$L(\mathbf{1})$-Labeling} is in XP for clique-width~\cite{SUCHAN20071885}, which implies that it is in XP for tree-width. 
Hanaka et al. also show that \textsc{$L(1,1)$-Labeling} is fixed-parameter tractable when parameterized by twin cover number \cite{Hanaka2022}.

The generalized setting, $L(\bp)$, is also studied but is less popular. Bertossi and Pinotti present approximation algorithms of \textsc{$L(\bp)$-Labeling} for trees and interval graphs \cite{bertossi2007approximate}. \textsc{$L(\bp)$-Labeling} is fixed-parameter tractable for the neighborhood diversity, $p_{\max}$, plus $k$~\cite{fiala2018parameterized}.
Further related work for \textsc{$L(\bp)$-Labeling} can be found in  the following surveys~\cite{Calamoneri2011:survey,Hasunuma2014:survey}.

\subsubsection{\textsc{(Metric Path) TSP}}
\textsc{Traveling Salesman Problem (TSP)} might be the most studied combinatorial optimization problem from both practical and theoretical points of view. Thus, we here list only a few of the results. 
 
On the practical side, an enormous number of works have been devoted to developing efficient algorithms for TSP for a long time. For example, as mentioned before, implementations of the Lin-Kernighan type algorithms~\cite{lin1973effective} have outstanding performance, and it was reported even in 2003~\cite{applegate2003chained} that an implementation of the chained Lin-Kernighan can constantly find near-optimal solutions for instances with 100,000 cities or more. Moreover, some implementations, such as Concorde and LKH~\cite{Concorde,LKH}, are available on the Web. Developments are continuing, and improvements are still reported~\cite{Tinos2018,TAILLARD2019420}. 

On the theoretical side, TSP has been studied from various aspects. For example, the Held-Karp algorithm with time complexity $O^*(2^n)$ was proposed in 1962~\cite{Bellman1962:TSP:DP,held1962dynamic}, and the existence of an exact algorithm with time complexity $O^*(c^n)$ for some $c < 2$ is still open~\cite{Woeginger2003}. 
For approximation, the general symmetric TSP has no approximation algorithm unless P$=$NP, whereas the \textsc{Metric TSP}, which is a restricted version of TSP whose edge-weights satisfy the triangle inequality, is known to be 1.5-approximable by the Christofides algorithm~\cite{Christofides1976:TSP}. 
Recently, this bound has been slightly improved by a randomized algorithm whose approximation ratio is at most $1.5-10^{-36}$~\cite{Karlin2021STOCTSP}.
Note that our reduction is not to \textsc{Metric TSP} but to \textsc{Metric Path TSP}.  Naive applications of algorithms for \textsc{Metric TSP} to \textsc{Metric Path TSP} do not preserve approximation guarantees, though it is shown that $\alpha$-approximation algorithm for \textsc{TSP} can be used to obtain an $(\alpha+\varepsilon)$-approximation solution of \textsc{Path TSP} for arbitrary $\varepsilon>0$~\cite{Traub2022SICOMP}. For \textsc{Metric Path TSP}, Zenklusen recently gives a deterministic $1.5$-approximation algorithm~\cite{zenklusen20191}. 
By combining the results on \cite{Karlin2021STOCTSP} and \cite{Traub2022SICOMP}, a randomized algorithm can obtain an approximate solution whose ratio is slightly better than 1.5.


\section{Preliminaries}
\subsection{Notations}
Let $G=(V,E)$ be an undirected and connected graph where $n=|V|$ and $m=|E|$.
The \emph{distance} between two vertices $u,v$ in $G$ is denoted by $\mathrm{dist}_{G}(u,v)$. The \emph{diameter} of $G$ is defined by $\mathrm{diam}\left(G\right)=\max_{u,v\in V}  \mathrm{dist}_{G}(u,v)$.
For a vertex $v\in V$, we denote by $N_G(v) = \left\{ u\in V\  \middle| \  \{u,v\}\in E  \right\}$ the set of adjacent vertices of $v$ in $G$.
For a vertex subset $S\subseteq V$, $G[S]$ is defined as the subgraph induced by $S$.
The complement graph of $G$ is denoted by $\overline{G}$. Also, the $k$-th power of graph $G$ is denoted by $G^k$.
Given a positive integer $k$, we define $[k]=\{ 1,2,\ldots,k \}$.
For an integer vector $\bp=(p_1,\ldots ,p_k)$, we define $p_{\min} = \min \{ p_1,\ldots , p_k\}$ and $p_{\max} = \max \{ p_1,\ldots , p_k\}$.
Let $\mathbf{1}=(1,\ldots,1)$ be a vector such that each element is $1$.


\subsection{Graph parameters}
A vertex subset $M \subseteq V$ is a \emph{module} of a graph $G$ if any pair of $u,v$ in $M$ satisfies that $ N_G(u) \setminus M = N_G(v) \setminus M $.

\begin{definition}[Modular-width]
A graph $G=(V,E)$ has \emph{modular-width} at most $\ell$ ($\ge 2$) if it satisfies (i) $|V|\le \ell$, or (ii) there is a partition $(V_1, \ldots, V_{\ell})$ of $V$ such that for each $i\in [\ell]$, $V_i$ is a module and $G[V_i]$ has \emph{modular-width} at most $\ell$. The minimum $\ell$ such that $G$ has modular-width at most $\ell$ is denoted by $\mw(G)$.
\label{def:MW}
\end{definition}


There is a polynomial-time algorithm that computes $\mw(G)$ and its decomposition \cite{tedder2008simpler}.
\begin{definition}[Neighborhood diversity]
A graph $G=(V,E)$ has \emph{neighborhood diversity} at most $\ell$ if there is a partition $(V_1,\ldots,V_{\ell})$ of $V$ such that every pair of vertices $u,v$ in $V_i$ satisfies $N_G(u)\setminus \{v\} = N_G(v)\setminus \{u\}$ for each $i\in [\ell]$.
The minimum $\ell$ is denoted by $\nd(G)$.
\label{def:ND}
\end{definition}
Note that each of $V_i$'s in Definition \ref{def:ND} is a module of $G$ and it forms either an independent set or a clique.
As with modular-width, there is a
polynomial-time algorithm for computing $\nd(G)$ and its partition \cite{Lampis2012:nd}. 

\begin{proposition}\label{prop:mw:complement}
For any graph $G=(V,E)$, $ \mw(G) = \mw(\overline{G})$ holds. 
\end{proposition}
\begin{proof}
It is sufficient to show that if $G$ has modular-width at most $\ell$, then $\overline{G}$ has modular-width at most $\ell$.
We show this claim by induction on the number of vertices $n$.
First, if $n \leq \ell$, then  both $G$ and $\overline{G}$ clearly satisfy condition (i), so the claim holds.

Next, assume that $n > \ell$ and that the claim holds for any graph whose number of vertices is less than $n$.
Let $(V_1,\ldots, V_{t})$ be a partition of $V$ such that each $V_i$ is a module and $G[V_i]$ has modular-width at most $\ell$. Note that $t\le \ell$.
Then, for each pair of $u,v \in V_i$, it holds that:
\begin{align*}
N_{\overline{G}}(u)\setminus V_i
&= (V\setminus N_{G}(u))\setminus V_i\\
&= (V\setminus N_{G}(v))\setminus V_i\\
&=N_{\overline{G}}(v)\setminus V_i.
\end{align*}
Therefore, $V_i$ is module of $\overline{G}$.
Furthermore, since $\overline{G}[V_i] =\overline{G[V_i]}$, 
$\overline{G}[V_i]$ has modular-width at most $\ell$ by the assumption of induction.
Therefore, $\overline{G}$ satisfies condition (ii) of Def.\ref{def:MW}.
\end{proof}

\begin{proposition}\label{prop:mw-nd}
For any connected graph $G=(V,E)$, $ \nd(G^2)\le \mw(G)$ holds, where $G^2$ is the second power of $G$.
\end{proposition}
\begin{proof}
If $|V|\le \mw(G)$, we are done as $\nd(G^2)\le |V|$.
Otherwise, consider a partition $(V_1,\ldots, V_{\ell})$ of $V$ such that $V_i$ is a module for each $i\in [\ell]$ where $\ell\le \mw(G)$.   Since $G$ is connected, any module is adjacent to at least one module, and vertices between two modules are completely joined;  that is, for the two modules $V_i$ and $V_j$, there is an edge $\{u,v\}$ between any pair of $u\in V_i$ and $v\in V_j$.
Thus, the distance of each pair of vertices in a module is at most 2, and hence each module forms a clique in $G^2$. Furthermore, for each pair of $u,v\in V_i$, $N_{G^2}(u)\setminus V_i = N_{G^2}(v)\setminus V_i$ follows from $N_G(u)\setminus V_i = N_G(v)\setminus V_i$. Therefore, $N_{G^2}(u)\setminus \{v\} = N_{G^2}(v)\setminus \{u\}$ holds, which implies  $ \nd(G^2)\le \mw(G)$.
\end{proof}


Finally, we introduce the clique-width $\cw(G)$ of $G$, which is a more general graph parameter than tree-width, modular-width, and neighborhood diversity. Namely, if some problem is not in FPT for tree-width, modular-width or neighborhood diversity, it is also not in FPT for clique-width. It is defined by some tree structures like tree-width, but we omit the detailed definition in this paper. Clique-width is a well-studied graph parameter, and many results are known. For example, cographs are the graph class of clique-width at most 2. We refer readers to \cite{Courcelle2000}.
In order to show the W[1]-hardness of \textsc{$L(2,1)$-Labeling} on graphs with diameter 2 when parameterized by clique-width  in Section \ref{sec:others}, we prove that \textsc{Hamiltonian Path} is W[1]-hard.

\begin{theorem}
\textsc{Hamiltonian Path} is W[1]-hard for clique-width.
\end{theorem}
\begin{proof}
We reduce \textsc{Hamiltonian Cycle}, which is W[1]-hard for clique-width \cite{Fomin2010:cw}. Given a graph $G=(V,E)$ of clique-width $\cw(G)$, pick arbitrary vertex $v$ and add a new vertex $v'$ that is adjacent to vertices in $N(v)$. That is, $v$ and $v'$ are false twins. Then we further add two vertices $w, w'$ that are adjacent to $v$ and $v'$, respectively. It is easily seen that $G$ has a hamiltonian cycle if and only if the constructed graph $G'$ has a hamiltonian path from $w$ to $w'$.
Since adding a vertex that is a false twin for some vertex to $G$ does not change the clique-width and adding a leaf vertex increases the clique-width by at most $2$,  $\cw(G')\le \cw(G)+4$ holds. This completes the proof.
\end{proof}

\section{Main results}\label{sec:results}
In this section, we show a polynomial-time reduction from  \textsc{$L(\bp)$-Labeling} to \textsc{Metric Path TSP}. 
\textsc{Path TSP} is the problem to finding a hamiltonian path of minimum weight on an edge-weighted complete graph.
Furthermore, \textsc{Metric Path TSP} is the restricted version of \textsc{Path TSP} such that the edge-weights of the input graph satisfy the triangular inequality.

\begin{theorem}\label{thm:DCLPtoTSP}
If $p_{\max}\leq 2p_{\min}$, \textsc{$L(\bp)$-Labeling} on graphs of diameter at most $k$ can be reduced to \textsc{Metric Path TSP} in $O(nm)$ time.
\end{theorem}

\begin{proof}
First, we define an edge-weighted complete graph $H=(V,\binom{V}{2},w)$ from an input graph $G$ (see Figure \ref{fig:fixed_order_labeling}).
For a pair of vertices $u,v\in V$ with $\mathrm{dist}_{G}(u,v)=d$, the edge weight of $\{u,v\}$ in $H$ is defined by $w(u,v)=p_d$.
Note that since $\mathrm{diam}(G)\leq k$, $w(u,v)$ is well-defined. Furthermore,   $p_{\min}\leq w(u,v) \leq 2p_{\min}$ holds by $p_d\leq 2 p_{\min}$ for each $d\in [k]$, and thus $w$ satisfies the  triangle inequality.

\begin{figure}[tb]
\centering
\includegraphics[keepaspectratio,width=0.8\hsize]{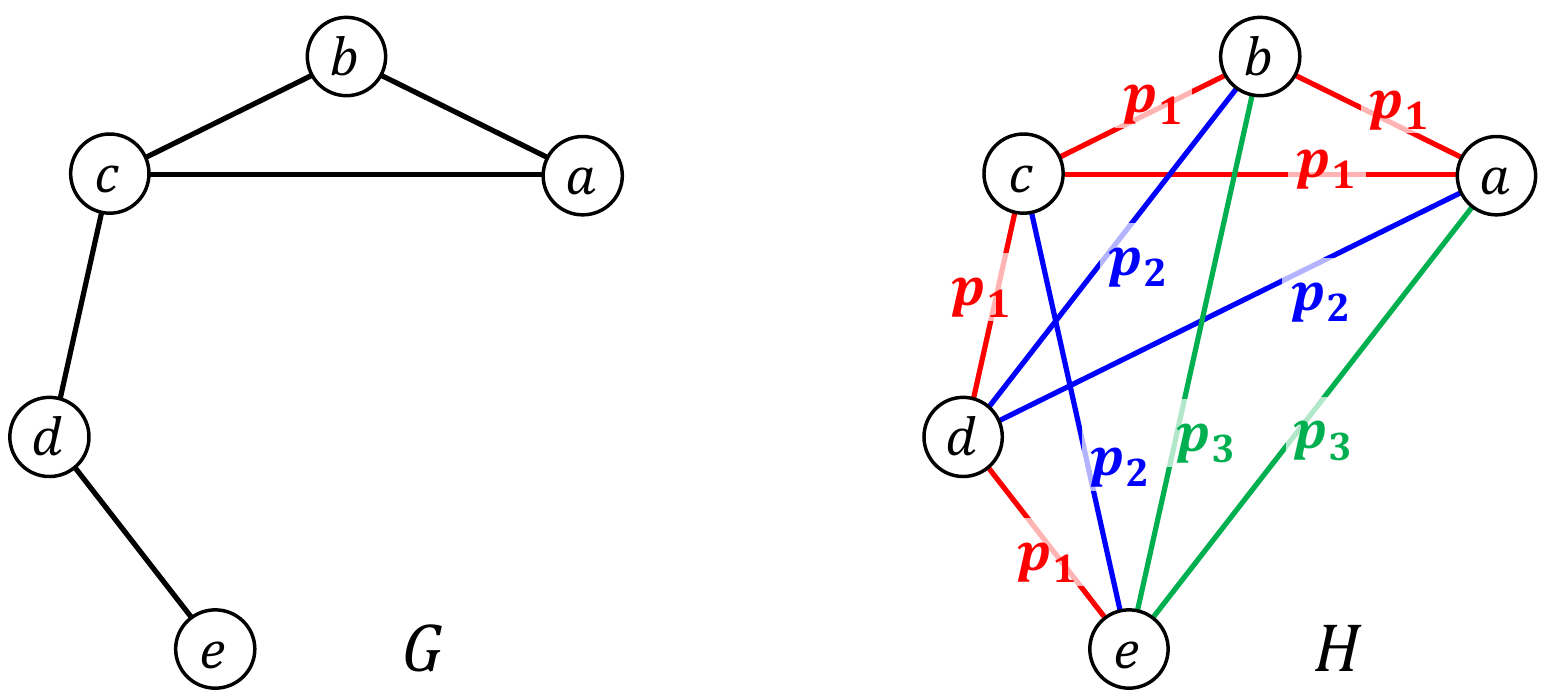}
\caption{The construction of $H$ for \textsc{$L(p_1,p_2,p_3)$-Labeling} on $G$ with diameter $3$.} \label{fig:fixed_order_labeling}
\end{figure}



For a permutation $\pi\colon V\to [n]$, we say that an $L(\bp)$-labeling $l$ is an \emph{$L(\bp)$-labeling for $\pi$} if it satisfies $l\left(\pi^{-1}(1)\right)\leq l\left(\pi^{-1}(2)\right)\leq \cdots\leq l\left(\pi^{-1}(n)\right)$. We denote by $\lambda_{\bp}(G,\pi)$ the minimum span among all of $L(\bp)$-labelings for $\pi$.
Here, we observe that any minimum $L(\bp)$-labelings for $\pi$ satisfies $l\left(\pi^{-1}(1)\right)=0$. If not, we obtain another labeling $l'$ such that $l'((\pi^{-1}(i))=l((\pi^{-1}(i))-1$, which contradicts the minimality of $l$.


Given a permutation $\pi$, let $l$ be an $L(\bp)$-labeling for $\pi$ with minimum span $\lambda_{\bp}(G,\pi)$ on $G$.
In the following, we denote $v_i=\pi ^{-1}(i)$ and $w_{i,j}=w(v_{i},v_{j})$ for simplicity. Then we show the following key claim, which implies that $l\left(v_i\right)$ is the length (sum of weights) of path $(v_1,v_2,\cdots,v_i)$ on $H$.

\begin{claim}\label{claim:induction}
For the edge-weighted complete graph $H$, the labeling $l$ satisfies that $l\left(v_i\right)=\sum_{t=1}^{i-1}w_{t,t+1}$ for any $i\in[n]$.
\end{claim}

\begin{proof}
We prove the claim by induction on $i$.
As the base case,  we have that $l(v_1)=0$. 
Furthermore, we consider the case of $i=2$. Since $l$ is a minimum $L(\bp)$-labeling for $\pi$, it satisfies that $0=l\left(v_1\right)\leq l\left(v_2\right)\leq \cdots\leq l\left(v_n\right)$. Since $0=l\left(v_1\right)\leq l\left(v_2\right)$, we have $l\left(v_2\right)\ge l\left(v_1\right)+p_{\dist_G(v_1,v_2)} = w_{1,2}$. Moreover, $l(v_2)\le w_{1,2}$ follows from $l\left(v_2\right)\leq \cdots\leq l\left(v_n\right)$ and the minimality of $l$. Thus, the claim holds when $i=2$.

In the induction step, assume that the claim holds for each $j\in [i-1]$.
By the minimality of $l$ and $l\left(v_1\right)\leq \cdots\leq l\left(v_n\right)$, the label of $v_{i}$ can be expressed as: 
\begin{align*}
l\left(v_{i}\right)
&=\min \left\{x \mid x\geq l\left(v_j\right)+p_{\dist_G(v_j,v_{i})}, \forall j\in [i-1]\right\}\nonumber \\
&=\min \left\{x \mid x\geq l\left(v_j\right)+w_{j,i}, \forall j\in [i-1]\right\}\nonumber \\
&=\max_{j \in [i-1]} \left\{ l\left(v_j\right)+w_{j,i} \right\}. \nonumber
\end{align*}
For each $j\in [i-2]$, it holds that
\begin{align*}
l\left(v_{i-1}\right)-l\left(v_j\right)
&=\sum_{t=1}^{i-2}w_{t,t+1}-\sum_{t=1}^{j-1}w_{t,t+1} \\
&=\sum_{t=j}^{i-2}w_{t,t+1}\\
&\geq w_{i-2,i-1}\geq p_{\min}.
\end{align*}
Furthermore, $w_{i-1,i}-w_{j,i}\geq p_{\min}-2p_{\min}=-p_{\min}$ holds.
Thus, for any $j\in [i-2]$, we have: 
\begin{align*}
    & (l\left(v_{i-1}\right)+w_{i-1,i})-(l\left(v_j\right)+w_{j,i})\\
&=(l\left(v_{i-1}\right)-l\left(v_j\right))+(w_{i-1,i}-w_{j,i})\\
&\geq p_{\min} -p_{\min} = 0.
\end{align*}
Consequently, we obtain 
\begin{align*}
   l\left(v_{i}\right) 
   &= \max_{j \in [i-1]} \left\{ l\left(v_j\right)+w_{j,i+1} \right\}\\
   &=l\left(v_{i-1}\right)+w_{i-1,i}\\
   &=\sum_{t=1}^{i-1}w_{t,t+1}.
\end{align*}

\end{proof}

Claim \ref{claim:induction} means that $\lambda_{\bp}(G,\pi)=l(v_n)$ is equivalent to the length of the hamiltonian path $\pi$ on $H$. Since $\lambda_{\bp}(G)=\min_{\pi} \left\{\lambda_{\bp}(G,\pi)\right\}$,  \textsc{Path TSP} on $H$ is equivalent to  \textsc{$L(\bp)$-Labeling}  on $G$. 


Finally, we discuss the running time of the reduction.
For the construction of $H$, we create the distance matrix of $G$. This can be done in $O(nm)$ time by the breadth-first search for each vertex. We then construct the weighted adjacency matrix of $H$ from the distance matrix of $G$. Clearly, it can be constructed in $O(n^2)$ time.
Thus, the total running time is $O(nm)+O(n^2)=O(nm)$.
\end{proof}
\begin{figure}[htb]
\centering
\includegraphics[keepaspectratio,width=\hsize]{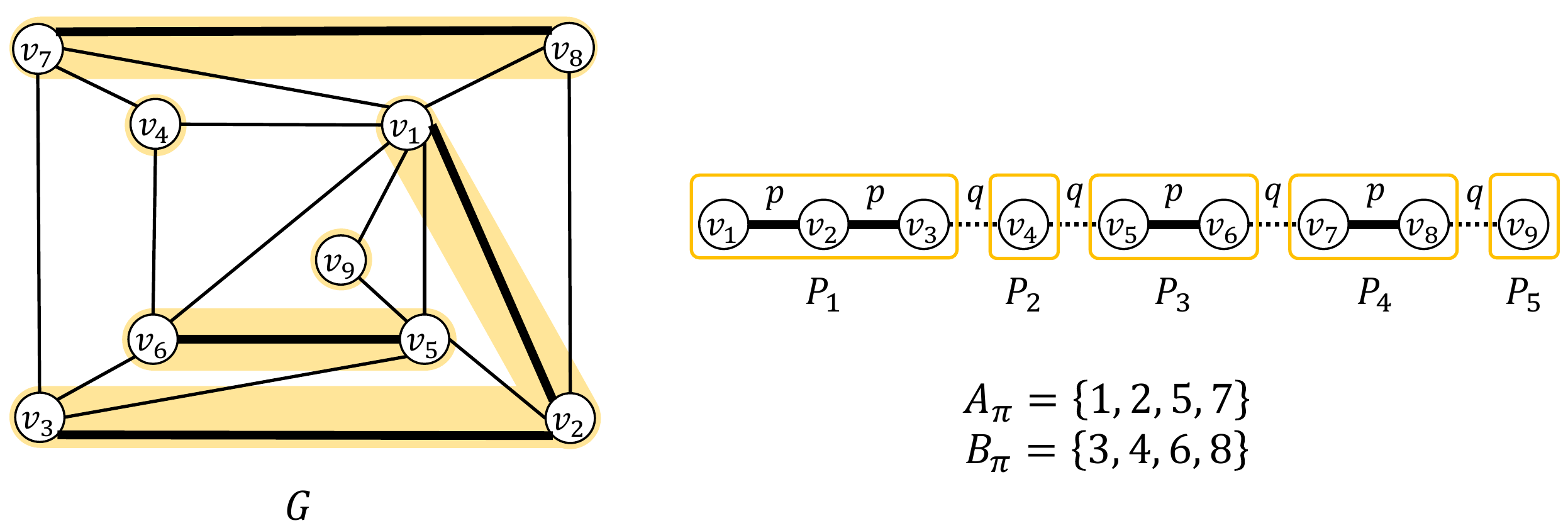}
\caption{Paths $P_1, \ldots, P_5$ consisting of only edges of weight $p$ along $\pi$ correspond to paths in $G$.} \label{figaar:L(pq)FPTmw}
\end{figure}

As a corollary, we can obtain an optimal solution in $O(2^nn^2)$ time and a $1.5$-approximate solution in polynomial time by applying algorithms for \textsc{Metric Path TSP} proposed in \cite{held1962dynamic} and \cite{zenklusen20191}, respectively, after the above reduction.
\begin{corollary}\label{cor:ExactExpo}
If $p_{\max}\leq 2p_{\min}$, \textsc{$L(\bp)$-Labeling} on graphs of diameter at most $k$ can be solved in $O(2^nn^2)$ time. Furthermore, it is approximable within $1.5$ in polynomial time.
\end{corollary}


Further observation shows that our problem is fixed-parameter tractable for modular-width.
\begin{corollary}\label{cor:mw}
The \textsc{$L(p,q)$-Labeling} problem on graphs of diameter at most 2 is fixed-parameter tractable for modular-width.
\end{corollary}
\begin{proof}
Let $G$ be a graph of diameter at most 2 and $H$ be the weighted complete graph obtained from $G$ as in Theorem \ref{thm:DCLPtoTSP}.
Notice that the weight of an edge in $H$ is either $p$ or $q$.

First, we consider the case that $p\leq q$.
For a permutation $\pi$ of $V$, we define: 
\begin{align*}
    A_{\pi}&=\left\{ i\in[n-1] \mid w_{i,i+1}=p \right\}\\  B_{\pi}&=\left\{ i\in[n-1] \mid  w_{i,i+1}=q \right\}.
\end{align*}
Note that $\{\pi^{-1}(i),\pi^{-1}(i+1)\}$ for $i\in A_{\pi}$ corresponds to an edge in $E$.

Since the weight of an edge in $H$ is either $p$ or $q$, the following equation holds: 
\begin{align*}
    \lambda_{\bp}(G,\pi)
&=\sum_{i=1}^{n-1}w_{i,i+1}
=\sum_{i\in A_{\pi}}p +\sum_{i\in B_{\pi}}q\\
&=(n-1)p+(q-p)\left|B_{\pi}\right|.
\end{align*}
Therefore, we have
$\lambda_{\bp}(G)=(n-1)p + (q-p)\min_{\pi}  \left|B_{\pi}\right|$. Since $n, p, q$ are constant,  solving \textsc{$L(\bp)$-Labeling} for $G$ is equivalent to finding $\pi$ that minimizes $\left|B_{\pi}\right|$ on $H$.

Here, let $P_1, \ldots, P_{s}$ be paths along $\pi$ such that each $P_i$ contains only edges with weight $p$ (see Figure \ref{figaar:L(pq)FPTmw}). Note that some $P_i$ could be one vertex. By the definition of such paths, $s=\left|B_{\pi}\right|+1$. We observe that edges in $P_i$ corresponds to edges in $G$. Thus, minimizing $\left|B_{\pi}\right|$ on $H$ is equivalent to the \textsc{Partition into Paths} problem, which is the problem to minimize the number of paths that partition $V$ in $G$. This can be computed in $f(\mw(G)) n^{O(1)}$ time \cite{gajarsky2013parameterized}.



For the case that  $p>q$, we can similarly solve \textsc{$L(p,q)$-Labeling} by computing \textsc{Partition into Paths} on the complementary graph $\overline{G}$ of $G$.
Since $\mw(\overline{G}) = \mw(G)$ by Proposition \ref{prop:mw:complement}, it can also be computed in $f(\mw(G)) n^{O(1)}$ time.
\end{proof}

\section{Related results}\label{sec:others}
In the previous section, we showed that \textsc{$L(p,q)$-Labeling} is fixed-parameter tractable for modular-width on graphs of diameter $2$. In this section, we first point out that \textsc{$L(2,1)$-Labeling} is  W[1]-hard for clique-width even on graphs of diameter $2$.

\begin{theorem}\label{thm:cw:hard}
\textsc{$L(2,1)$-Labeling} on graphs with diameter 2 is W[1]-hard for clique-width.
\end{theorem}
\begin{proof}
In \cite{griggs1992labelling}, Griggs and Yeh give a reduction from \textsc{Hamiltonian Path} to \textsc{$L(2,1)$-Labeling} on graphs with diameter 2. Given a graph $G=(V,E)$ as an instance of \textsc{Hamiltonian Path}, the reduced graph of \textsc{$L(2,1)$-Labeling} is constructed by taking the complementary graph $\overline{G}$ of $G$ and adding a universal vertex $x$ that is adjacent to all the vertices in $V$.
Since $\cw(\overline{G})\le 2\cw(G)$ holds for any graph $G$ \cite{Courcelle2000} and adding a universal vertex $x$ increases the clique-width of $\overline{G}$ by at most 1, the clique-width of the reduced graph in \cite{griggs1992labelling} is at most $2\cw(G)+1$. This completes the proof.
\end{proof}
Note that \textsc{$L(1,1)$-Labeling} on graphs with diameter 2 is trivially solvable because the graph power $G^2$ of a graph of diameter 2 is a complete graph.

The fixed-parameter tractability of \textsc{$L(p,q)$-Labeling} for modular-width remains open in general.
On the other hand, we show that \textsc{$L(1,1)$-Labeling} and even \textsc{$L(\mathbf{1})$-Labeling} on general graphs are fixed-parameter tractable by modular-width in contrast to \textsc{$L(p,q)$-Labeling}. 

\begin{theorem}\label{thm:L11:mw}
\textsc{$L(\mathbf{1})$-Labeling} on general graphs is fixed-parameter tractable for modular-width.
\end{theorem}
\begin{proof}
As mentioned in \cite{fiala2018parameterized}, $\nd(G)\ge \nd(G^k)$ holds for any graph $G$ and any positive integer $k\ge 1$. By Proposition \ref{prop:mw-nd}, we have $\mw(G)\ge \nd(G^2)\ge \nd(G^k)$ for any positive integer $k\ge 2$. Also, \textsc{$L(\mathbf{1})$-Labeling} on $G$ is equivalent to \textsc{Coloring} on $G^k$. We know that \textsc{Coloring} is fixed-parameter tractable for neighborhood diversity \cite{Lampis2012:nd}.
Solving \textsc{Coloring} on $G^k$, one can compute \textsc{$L(\mathbf{1})$-Labeling}  in  $f(\mw(G))n^{O(1)}$ time.
\end{proof}

As the corollary of Theorem \ref{thm:L11:mw}, we obtain an FPT-approximation algorithm for \textsc{$L(\bp)$-Labeling} with respect to modular-width.
\begin{corollary}\label{thm:Lpq:mw}
There is a $p_{\max}$-approximation fixed-parameter algorithm for \textsc{$L(\bp)$-Labeling} on general graphs with respect to modular-width.
\end{corollary}
\begin{proof}
For any constant $c$, $\lambda_{c\bp} = c\lambda_{\bp}$ holds. Thus, we have $\lambda_{\bp}\le \lambda_{p_{\max}\mathbf{1}}\le p_{\max}\lambda_{\mathbf{1}}$. By Theorem \ref{thm:L11:mw}, we obtain a $p_{\max}$-approximation fixed-parameter algorithm by modular-width.
\end{proof}



\bibliographystyle{IEEEtranS}
\bibliography{ref}
\end{document}